\documentclass[journal, onecolumn, 11pt]{IEEEtran}
\usepackage[onehalfspacing]{setspace}
\usepackage{amsmath}
\usepackage{amssymb}
\usepackage{mathrsfs}
\usepackage{amsthm}
\usepackage{multirow}
\usepackage{color}
\usepackage{array}
\usepackage{url}
\usepackage{comment}
\usepackage{enumerate}
\usepackage{eucal}
\usepackage{bbm}
\usepackage{cleveref}

\usepackage{tikz}
\usetikzlibrary{shapes,arrows,automata}
\usetikzlibrary{arrows, arrows.meta,calc,positioning,shapes.arrows}
\usetikzlibrary{decorations.pathreplacing}
\usetikzlibrary{chains,fit,shapes}
\usetikzlibrary{mindmap,trees}

\tikzset{
	ncbar angle/.initial=90,
	ncbar/.style={
		to path=(\tikztostart)
		-- ($(\tikztostart)!#1!\pgfkeysvalueof{/tikz/ncbar angle}:(\tikztotarget)$) 
		-- ($(\tikztotarget)!($(\tikztostart)!#1!\pgfkeysvalueof{/tikz/ncbar angle}:(\tikztotarget)$)!\pgfkeysvalueof{/tikz/ncbar angle}:(\tikztostart)$) 
		\tikztonodes
		-- (\tikztotarget) 
	},
	ncbar/.default=0.5cm,
}

\usepackage{algorithm}
\usepackage[noend]{algorithmic}
\usepackage{cite}

\usepackage{stmaryrd}

\usepackage[top=0.75in, bottom=0.75in, left=0.75in, right=0.75in]{geometry}

\usepackage{tikz}
\usetikzlibrary{arrows, patterns, shapes.arrows, decorations.pathmorphing}

\IEEEoverridecommandlockouts

\theoremstyle{plain}
\newtheorem{theorem}{Theorem}
\newtheorem{corollary}[theorem]{Corollary}

\newtheorem{proposition}[theorem]{Proposition}

\theoremstyle{definition}
\newtheorem{definition}{Definition}
\newtheorem{example}[definition]{Example}

\newtheorem{construction}{Construction}

\newcommand{\C}{\mathtt{C}}

\DeclareMathAlphabet{\mathbfsl}{OT1}{ppl}{b}{it} %{OT1}{cmr}{bx}{it}

\newcommand{\bx}{\mathbfsl{x}}
\newcommand{\by}{\mathbfsl{y}}
\newcommand{\bz}{\mathbfsl{z}}
\newcommand{\B}{\mathcal{B}}

\newcommand{\vA}{\mathbfsl{A}}
\newcommand{\vB}{\mathbfsl{B}}
\newcommand{\vC}{\mathbfsl{C}}

\newcommand{\vH}{\mathbfsl{H}}

\newcommand{\cS}{\mathcal{S}}

\newcommand{\cC}{\mathcal{C}}

\newcommand{\bc}{\mathbf{c}}

\newcommand{\etal}{{\em et al.}}

\newcommand{\floor}[1]{{\left\lfloor #1\right\rfloor}}
\newcommand{\ceil}[1]{{\left\lceil #1\right\rceil}}

\title{Deletion\,Correcting\,Codes~for~Efficient\,DNA\,Synthesis\\[0mm]}

 \author{{Johan Chrisnata},\!\IEEEauthorrefmark{1}
	 	{Han Mao Kiah},\!\IEEEauthorrefmark{1}
	 	and {Van Long Phuoc Pham}\IEEEauthorrefmark{1}\\[1mm]
	 	\IEEEauthorblockA{\IEEEauthorrefmark{1} \footnotesize School of Physical and Mathematical Sciences, 
		 		Nanyang Technological University, Singapore\\[0mm]}
	 	{  \{johan.chrisnata,hmkiah,phuoc002\}@ntu.edu.sg.edu}\\[-4mm]}
\date{}

\begin{document}
	\maketitle
	
	\hspace*{-14pt}
	\begin{abstract}
		The synthesis of DNA strands remains the most costly part of the DNA storage system. 
		Thus, to make DNA storage system more practical, the time and materials used in the synthesis process have to be optimized. We consider the most common type of synthesis process where multiple DNA strands are synthesized in parallel from a common alternating supersequence, one nucleotide at a time. 
		The synthesis time or the number of synthesis cycles is then determined by the length of this common supersequence. 
		In this model, we design quaternary codes that minimizes synthesis time that can correct deletions or insertions, which are the most prevalent types of error in array-based synthesis. 
		We also propose polynomial-time algorithms that encode binary strings into these codes and show that the rate is close to capacity.
		%quaternary codewords with synthesis time restriction that can correct deletions or insertions whose rate is close to the capacity using dynamic programming.
	\end{abstract}

	\setstretch{1.05}

	\section{Introduction}\label{sec:intro}
	DNA storage has been progressing rapidly in the last decade due to its viability in storing information, 
	namely its durability and high storage density (see \cite{Yazdi.etal:2015b,Shomorony.2022} and the references therein). 
	The process of storing information in the form of synthetic DNA strands involves converting original binary information into quaternary strings (nucleotide bases) and writing those information into DNA strands using the synthesis machine. 
	The corresponding DNA strands are then stored in DNA pools. 
	To read/retrieve the stored data, the user employs a sequencing platform that creates multiple possibly erroneous copies of the same strand which are clustered and decoded to the original binary strings.
	
	%\textbf{Why is it interesting? and what is the motivating problem}
	Although many works have been done on improving the feasibility of DNA storage, %\cite{Lenz2019,Newman2019,Rashtchian2017,Willsey2019}, 
	they are still far from being practical in storing large data, mainly because of its cost. 
	Specifically, the DNA synthesis remains the most costly part of the DNA storage model. 
	In order to make the synthesis part efficient and less costly, we need to understand how the synthesis process works.
	
	To limit the error in the synthesis process, typically each DNA strand contains no more than 250 nucleotides~\cite{Carmean2019}
	(see also Table 1.1 in~\cite{Shomorony.2022} for experiments for DNA data storage). 
	Thus, the user splits the converted quaternary strings into multiple short sequences to synthesize and stores them in a container in an unordered manner. The task of the synthesizing these multiple strands is usually array-based, which means that they are done in parallel \cite{Ceze2019}. To be exact, the synthesizer scans a fixed supersequence of DNA nucleotides and appends a nucleotide at a time to the corresponding strands. In other words, in each cycle of synthesis, the machine only appends one nucleotide to a subset of the DNA strands that require that particular nucleotide. As a result, all of the synthesized DNA strands must be subsequences of the fixed supersequence. The length of the fixed supersequence then determines the number of cycles required to generate all strands which corresponds to the total synthesis time. Our goal now is thus to minimize the total synthesis time which is consistent with decreasing the total cost of DNA synthesis since each cycle requires reagents and chemicals \cite{Ceze2019}.
	
	%\textbf{Why hasn't it been solved before or whats "wrong" with previous proposed solution?}
	Lenz \etal \cite{Lenz2020} showed that the supersequence that maximizes the information rate is the alternating quaternary sequence. Therefore, for this paper, we always use the alternating quaternary sequence as our supersequence. 
	For example, a naive scheme in which the binary information is encoded into DNA strings is using the rule $(00) \xrightarrow{}A$, $(01) \xrightarrow{}C$, $(10) \xrightarrow{}G$, $(11) \xrightarrow{}T$. 
	In this case, an alternating sequence which repeats the substring ACGT is scanned and the synthesis cycle begins. 
	This scheme achieves an information rate of 0.5 bits/cycle since four cycles are required to synthesize one nucleotide which is equivalent to two bits of information. 
	They then showed that by introducing redundancy to the synthesized strands using constrained coding, they can improve the information rate from 0.5 bits/cycle to 0.947 bits/cycle. 
	However, the encoding scheme does not correct errors that happen during synthesis.
	
	%\textbf{What is the problem?}
	%This paper aims to study codes constructions and encoding for efficient DNA synthesis in the presence of errors. 
In this paper, we design codes for efficient DNA synthesis in the presence of errors.
The most prevalent types of errors that stem from array-based synthesis are deletion and insertion \cite{Ceze2019}.
We use $\B$ to denote an error-function ball that maps a quaternary sequence $\bx$ to the set of its possible errors, $\B(\bx)$. To this end, we define an $(n,T,\B)$-synthesis code to be a quaternary code $\C$ of length $n$ which is a subset of the set of all subsequences of the alternating quaternary supersequence of length $T$, such that for any two codewords $\bx$ and $\by$ in $\C$, $\B(x) \cap \B(y)=\emptyset$. Observe that there is a trade off between the synthesis time $T$ or the length of the fixed supersequence with the rate of the code $\C$. 
Thus, we aim to construct optimal codes with high information rate per symbol with different synthesis time constraints that also corrects deletion or insertion errors. For this paper, we fix $\B$ to be a single insertion or deletion error ball (\textit{indel}), which means that $\B(\bx)$ is the set of all words that can be obtained by at most a single deletion or at most a single insertion from the word $\bx$.
In addition, we also design efficient encoders for our codes.
	
	To correct a single indel, we have the famous binary Varshamov-Tenengolts (VT) codes~\cite{Levenshtein1966}.
	This class of codes was later extended to the nonbinary case by Tenengolts~\cite{Tenengolts1984}, who also provided a linear-time systematic encoder. 
	However, in the worst case, a length-$n$ codeword of a quaternary VT-code may have synthesis time $4n$.
	Hence, in this paper, we fix a synthesis time $n< T \le 4n$ and study ways to encode into quaternary VT-codewords whose synthesis time is at most $T$.

	Before we formulate our problem, we mention some works that adapt VT codes for DNA based data storage. 
	In~\cite{Cai2021,NCIK2021}, the authors provided efficient encoders that map messages into VT codewords that obey certain GC-content and homopolymer constraints.
	\cite{Gabrys2017,NCIS2022} modified VT codes to correct adjacent transpositions, chromosome-inversions, 
	while \cite{CKNY2021, CKY2022} adapted VT codes to correct deletions using more reads.

	\begin{figure}
		\centering
		\includegraphics[width=10cm, scale=0.5]{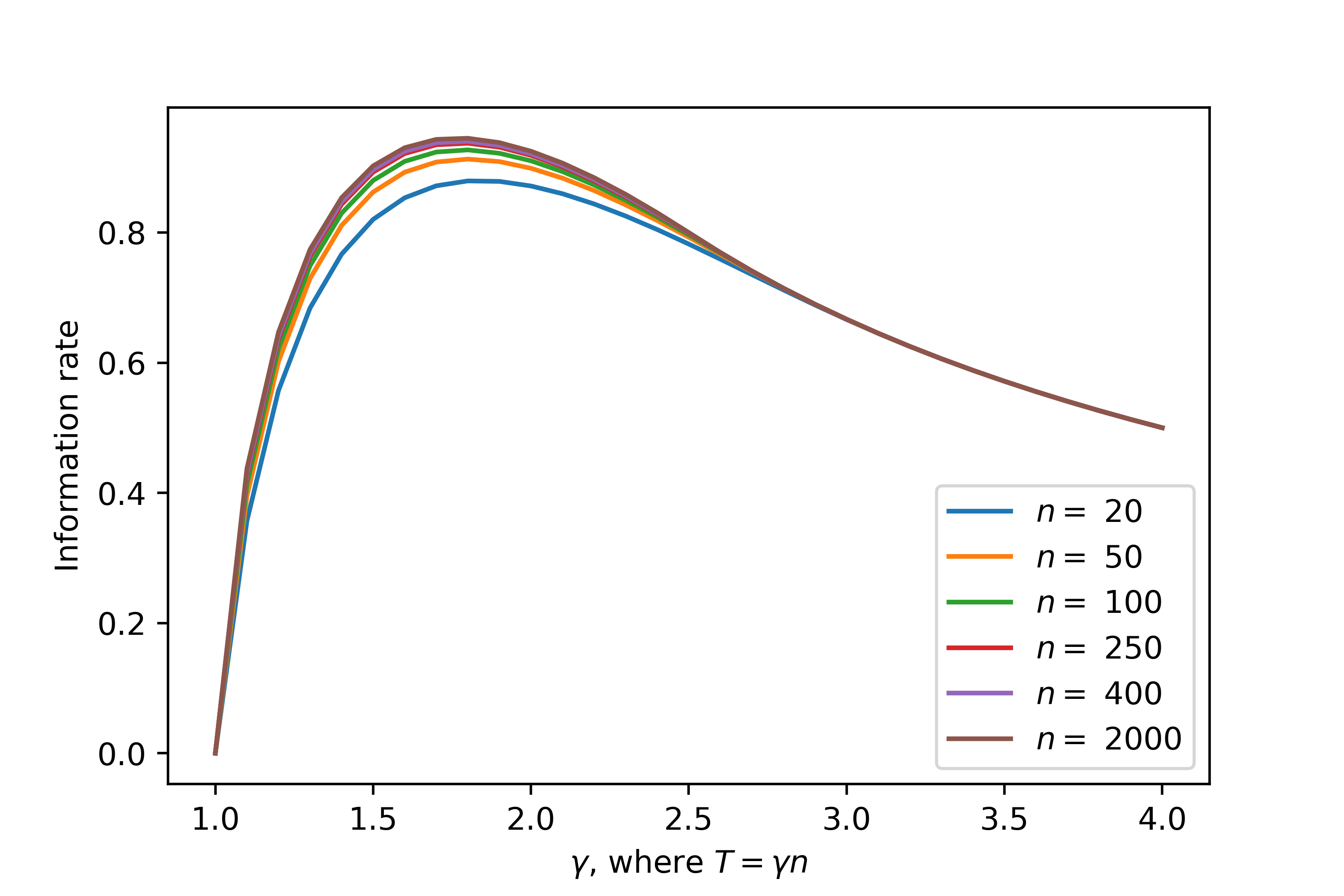}
		\caption{Plot for information rate for $W(n, \gamma n)$}
		\label{graph:information_rate}
	\end{figure}
	
	\section{Preliminaries}
	\label{sec:problem}
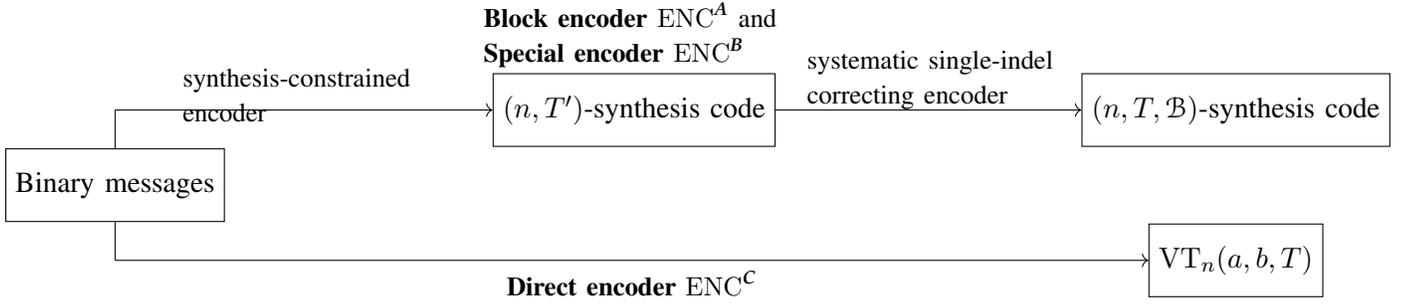
\begin{figure*}[!h]
	
	\begin{center}
		\begin{tikzpicture}%[scale=0.95]
			%	\footnotesize
			\tikzset{vertex/.style = {circle,fill, inner sep = 2pt}}
			\tikzset{edge0/.style = {->,> = latex',blue,thick}}
			\tikzset{edge1/.style = {->,> = latex',red,dashed,thick}}
			\node[state,rectangle] (B) at (5,3) {$(n,T')$-synthesis code}; % Here the nodes and coordinates are defined
			\node[state,rectangle] (C) at (13,3) {$(n,T,\B)$-synthesis code};
			\node[state,rectangle] (D) at (13,1) {${\rm VT}_n(a,b,T)$};
			\node[state,rectangle] (A) at (-1.9, 2) {Binary messages};
			\node[text width=4cm] at (5,4) {\small \textbf{Block encoder} ${\rm ENC}^\vA$ and \textbf{Special encoder} ${\rm ENC}^\vB$ };
			\node[text width=4cm] at (1,3.2) {\small synthesis-constrained encoder};
			\node[text width=4cm] at (9.3,3.4) {\small systematic single-indel correcting encoder};
			%\node[text width=4cm] at (6.5,1) {direct encoder to ${\rm VT}_n(a,b,T)$};
			
			\path[->] % path and draw commands connect the nodes and coordinates to each other.
			(-1.9,3) edge [] node  {} (B); 
			\path[->] (B) edge [] node  {} (C); 
			\path[-] (A) edge [] node  {} (-1.9,1); 
			\path[-] (A) edge [] node  {} (-1.9,3); 
			\path[->] (-1.9,1) edge [] node[pos=0.5,below]  {\small \textbf{Direct encoder} ${\rm ENC}^\vC$} (D); 
			%path[->] (A) edge [bend right] node  {} (C);
			%	\draw[->] (A) to [ncbar=-1cm] node[pos=0.5,below] {Direct encoder ${\rm ENC}^\vC$} (D); 
		\end{tikzpicture}
	\end{center}
	\caption{Single-indel correcting synthesis-constrained encoder}\label{fig:encoders}
	%\hm{Think about the diagram.}}
\end{figure*}
	Let $\Sigma=\{1,2,3,4\}$ be an alphabet of size four to represent the four nucleotides of DNA. 
	Throughout this paper, we use the shifted modulo operator $(a\bmod 4) \in \{1,2,3,4\}$. 
	Let $\Sigma^n$ denote the set of all quaternary words of length $n$ and $\Sigma^{n*}=\Sigma^{n-1} \cup \Sigma^n \cup \Sigma^{n+1}$. 
	For a positive integer $T$, let $\mathbf{a}_T=12341234\cdots$ denote the \textit{alternating quaternary sequence} of length $T$.
	% \hm{Was this term used in previous paper?} {\color{red}JC : In the efficient synthesis paper, they used $A_q$ to denote the alternating sequence that cycically repeats all symbols in $\Sigma$ in ascending order}. 
	For a quaternary word $\bx=x_1 x_2 \cdots x_n \in \Sigma^n$, we define the \textit{differential word of $\bx$} to be $D(\bx)\triangleq (x_1 \bmod 4 , x_2-x_1 \bmod 4, x_3-x_2 \bmod 4,\ldots,x_n-x_{n-1} \bmod 4)$. Note that the map $D:\Sigma^n \xrightarrow{} \Sigma^n$ is a one to one correspondence and hence is invertible. For a quaternary word $\bx \in \Sigma^n$, we define $||\bx||\triangleq \sum_{i=1}^{n}{x_i}$ to be the $L_1$-norm of $\bx$. We define the \textit{synthesis time of $\bx$}, denoted $S(\bx)$, with respect to the alternating quaternary sequence to be the smallest integer $T$ such that $\bx$ is a subsequence of $\mathbf{a}_T$. Furthermore, we can show that $S(x)=||D(\bx)||$.
	
	\begin{example}
		Let $\bx=23144$. Note that $\bx$ is a subsequence of $\mathbf{a}_{12}=123412341234$, and thus $S(\bx)=12$. On the other hand, $D(\bx)=21234$, and $||D(\bx)||=12=S(\bx)$. 
	\end{example}
	
	%\hm{Change the name of constrained.}{\color{red}JC : maybe I change it to synthesis code instead.}
	Let $\cC\subseteq \Sigma^n$. For $\alpha \in \Sigma$, we denote $\cC \circ \alpha=\{\bx \circ \alpha : \bx \in \cC\}$, where $\bx \circ \alpha$ is the word obtained by appending $\alpha$ to $\bx$. We say that $\cC$ ia an {\em $(n,T)$-synthesis code} if for every $\bx\in \cC$, we have $S(\bx)\leq T$. 
	Furthermore, $\cC$ is an $(n,T,\B)$-synthesis code if $\cC$ is an $(n,T)$-synthesis code and for every $\bx,\by \in \cC$, we have $\B(\bx) \cap \B(\by) =\emptyset$. 
	We define its {\em redundancy} to be $2n - \log_2 |\mathcal{C}|$ bits. 
	Let $W(n,T)$ denote the largest $(n,T)$-synthesis code and let $A(n,T)$ denote its size. 
	Likewise, we define $A(n,T,\B)$ to be the maximum size of an $(n,T,\B)$-synthesis code. 
	For $n<T\leq 4 n$, we define the {\em information rate} of an $(n,T)$-synthesis code $\cC$ to be $R(\cC,T)=\frac{\log_2 |\cC|}{T}$ bits per cycle. %Likewise, the {\em information rate} of an $(n,T,\B)$-synthesis code $\cC$ is $R(n,T,\B)=\frac{\log_2 |\cC|}{T}$ bits per cycle. 
	Thus, for some $1< \gamma \leq 4$, 
	we define the corresponding {\em capacity} to be $C(\gamma)=\limsup_{n \to \infty}{\frac{\log_2 A(n,\gamma n)}{\gamma n}}$, 
	and the {\em asymptotic rate} $C(\gamma,\B)=\limsup_{n \to \infty}{\frac{\log_2 A(n,\gamma n,\B)}{\gamma n}}$.
	% bits per cycle respectively.
	
	Although there is no explicit closed formula for $A(n,T)$ for any value of $n$ and $T$, we can find its value using recursion (see Section~\ref{subsec:generalencoder} for details). 
	Nevertheless, the generating fuction for $A(n,T)$ was provided by Lenz \etal{} \cite{Lenz2020}.
	\begin{proposition}[Lenz \etal\cite{Lenz2020}]
		%The value of $A(n,T)$ can be obtained as the coefficient of $z^T$ of the generating function %$\frac{1}{1-z}(z+z^2+z^3+z^4)^n=\frac{z^n(1-z^4)^n}{(1-z)^{n+1}}$
		We have that 
		\[\sum_{T\ge 0} A(n,T)z^T = \frac{1}{1-z}(z+z^2+z^3+z^4)^n=\frac{z^n(1-z^4)^n}{(1-z)^{n+1}}\,.\]
	\end{proposition}
	
	The graph for the values of information rate of $W(n,T)$ for various values of $T=\gamma n$ is given in Figure \ref{graph:information_rate}.  
	
	Let $\B$ be a single deletion or a single insertion error (indel error), which means that $\B(\bx)$ is the set of all words that can be obtained from $\bx$ by at most one deletion or at most one insertion. Our interest is now to estimate the value of $A(n,T,\B)$, namely to give upper bound and lower bound of the value. In the next section, we present a construction of an $(n,T,\B)$-synthesis code for any $T$, where $n< T \leq 4n$. We also provide encoding algorithm/encoder to construct such codes. To be more formal, we define the following notion of encoder.
	
	\begin{definition}\label{def:constrainedencoder}
		The map ${\rm ENC} :\{0,1\}^m \xrightarrow{} \Sigma^n$ is a \textit{single-indel correcting synthesis-constrained encoder} if there exists a \textit{decoder} map ${\rm DEC} : \Sigma^{n*} \xrightarrow{} \{0,1\}^m$ such that
		\begin{enumerate}[(i)]
			\item For every $\bx \in \{0,1\}^m$, if $\bc={\rm ENC}(\bx)$ and $\by \in \B(\bc)$, then $\text{DEC}(\by)=\bx$,
			\item For every $\bx \in \{0,1\}^m$, $S({\rm ENC}(\bx)) \leq T$.
		\end{enumerate}
		This means that the code $\cC=\{{\rm ENC}(\bx): \bx \in \{0,1\}^m \}$ is an $(n,T,\B)$-synthesis code.
	\end{definition}
	
	In this paper, we provide several \textit{single-indel correcting synthesis-constrained encoders} that map binary strings to an $(n,T,\B)$-synthesis code. Specifically, for $T\geq 2.5 n$, the encoder in Algorithm \ref{alg:constrainedencoder} runs in linear time, while for general $T$, the encoders in Algorithm \ref{alg:constrainedencoderquad} and Algorithm~\ref{alg:directencoder} run in $O(n^3)$ and $O(n^5)$ time respectively.
	Before we describe these encoders, we describe the main ingredient of the construction -- the VT-codes.
	
	\subsection{Varshamov-Tenengolts (${\rm VT}$) codes}
	%Before we present our main results in Subsection~\ref{subsec:generalencoder} and~\ref{subsec:special}, we define several terminologies and past results that we are going to use throughout this section.
	\begin{definition}
		Let $\bx$ be a binary word of length $n$. Then the {\em VT-syndrome} of $\bx$ is defined as ${\rm VT}(\bx)=\sum_{i=1}^n{ix_i}$. 
	\end{definition}
	
	\begin{definition}
		Let $\bx \in \Sigma^n$ be a quaternary word of length $n$. We define the {\em auxiliary binary sequence} of 
		$\bx$ to be the binary sequence $\tilde{\bx}$ of length $n-1$ such that
		$\tilde{x}_i=
		\begin{cases}
			1 &\quad \text{if } x_{i+1} \geq x_i\\
			0 &\quad \text{if } x_{i+1} < x_i,\\
		\end{cases}$
		for all $1 \leq i\leq n-1$.
	\end{definition}
	\begin{construction}[Tenengolts\cite{Tenengolts1984}]
		Let $n$ be a positive integer. We define
		${\rm VT}_n(a,b)=\{ \bx \in \Sigma^n : {{\rm VT}}(\tilde{\bx})=a \bmod n,$ and $\sum_{i=1}^n{x_i}=b \bmod 4\}.$
	\end{construction}
	\begin{theorem}[Tenengolts\cite{Tenengolts1984}]
		For any $a \in \mathbb{Z}_n$ and $b \in \mathbb{Z}_4$, the set ${\rm VT}_n(a,b)$ is an $(n,4n,\B)$-synthesis code. Moreover, there exist $a \in \mathbb{Z}_n$ and $b \in \mathbb{Z}_4$ such that the redundancy of ${\rm VT}_n(a,b)$ is at most $\ceil{\log_2 n}+ 2$. Furthermore, there exists a systematic quaternary encoder that maps binary sequence to a quaternary code that can correct single insertion or single deletion error with redundancy at most $\log_2 n +12$.
	\end{theorem}

	In the same paper, Tenengolts also provided a linear-time decoder that corrects a single indel.
	
	\begin{theorem}[Tenengolts\cite{Tenengolts1984}]\label{thm:indeldecoder}
		Let $a \in \mathbb{Z}_n$ and $b \in \mathbb{Z}_4$. Given $\bx \in {\rm VT}_n(a,b)$, for any $\by \in \B(\bx)$, there exists a linear-time decoder ${\rm DEC}_{a,b}:\Sigma^{n*} \xrightarrow{} \Sigma^n$, such that ${\rm DEC}_{a,b}(\by)=\bx$.
	\end{theorem}
	
	Next, we give lower bound on the size of an $(n,T,\B)$-synthesis code. % as summarized in the theorem below.
	\begin{definition}
		We define the set
		${\rm VT}_n(a,b,T)\triangleq \{\bx \in \Sigma^n : S(\bx) \leq T, {\rm VT}(\tilde{\bx})=a \bmod n,$ and $\sum_{i=1}^n{x_i}=b \bmod 4\}.$
	\end{definition}
	\begin{theorem}\label{thm:VTabT}
		For $n<T\leq 4n$, we have that ${\rm VT}_n(a,b,T)$ is an $(n,T,\B)$-synthesis code with size at least $A(n,T)/4n$.
	\end{theorem}
		
\begin{proof}
		Consider the set of all quaternary words of length $n$, whose synthesis time is at most $T$. We know that there are $A(n,T)$ such words. Note that ${\rm VT}_n(a,b,T)$ is  a subset of ${\rm VT}_n(a,b)$, therefore $B(\bx) \cap \B(\by)=\emptyset$ for any $\bx,\by \in {\rm VT}_n(a,b,T)$. Thus ${\rm VT}_n(a,b,T)$ is an $(n,T,B)$-synthesis code. There are $4n$ choices of parameters ($4$ choices of $b$ and $n$ choices of $a$), so by the pigeon hole principle, there exist a choice of $a$ and $b$ such that $VT_n(a, b, T)$ has size at least $A(n, T) / 4n$.
	\end{proof}
	
	%	\end{comment}
	\begin{corollary}\label{cor:VTabT}
		For $2.5 n \leq T \leq 4n$, there exists an $(n,T,\B)$-synthesis code with size at least $\frac{4^{n-1}}{2n}$. Hence the redundancy of the code is at most $3+ \log_2 n$ bits.
	\end{corollary}

	Now, Theorem~\ref{thm:VTabT} and Corollary~\ref{cor:VTabT} are existence results.
	In other words, a greedy method that constructs codes meeting the lower bound takes exponential time.
	Therefore, in this work, we provide explicit construction of codes whose rates are close to the ones promised in Theorem~\ref{thm:VTabT} and Corollary~\ref{cor:VTabT}.
	Specifically, we present several ways to encode binary messages to an $(n,T,\B)$-synthesis code. 
	The encoders are summarized in Figure~\ref{fig:encoders}.
	\begin{enumerate}[(A)]
		\item \textbf{Block encoder ${\rm ENC}^\vA$}. %\hm{Why use $J,I,K$? Why not just $A,B,C$ to follow the subsections they appear?}{\color{red}JC : done}
		The first idea is to encode binary messages to an $(n,T')$-synthesis code, where $T=T'+O(\log_2 n)$. 
		We do this by splitting the binary messages into $\ell$ blocks and encode each block using ranking/unranking to the set of all quaternary strings of length $k$ with synthesis time at most $T'/\ell$. Finally we concatenate all the blocks and encode it to an $(n,T,\B)$-synthesis code using a \textit{systematic single-indel correcting code}, ${\rm ENC}^\vH$.
		The whole process is presented as a \textit{single-indel correcting synthesis-constrained encoder}, ${\rm ENC}^\vA$, in Algorithm~\ref{alg:constrainedencoderquad} that runs in $O(n^3)$ time. 
		This is discussed in Subsection~\ref{subsec:generalencoder}.
		
		\item \textbf{Special encoder ${\rm ENC}^\vB$}. For the special case where $T\geq 2.5 n$, we borrow ideas from  \cite{Lenz2020} to provide another encoder, ${\rm ENC}^\vB$, that incurs smaller redundancy (as compared to ${\rm ENC}^{\vA}$). This encoder is given in Algorithm~\ref{alg:constrainedencoder} and it runs in $O(n)$ time.
		
		\item \textbf{Direct encoder ${\rm ENC}^\vC$}. Finally we provide a method that encodes binary messages directly to a ${\rm VT}_n(a,b,T)$ code using ranking/unranking algorithm for any general $T$. This encoder ${\rm ENC}^\vC$ is given in Algorithm~\ref{alg:directencoder} and has the best redundancy among other encoders. However, it runs in $O(n^5)$ time.
	\end{enumerate}
	
	%We compare the rates of the three encoders in Figure~\ref{fig:threecomparison}.
	We compare the information rates of the three encoders in Table~\ref{tab:threecomparison}.
	
	\begin{comment}
	\begin{figure}
	\centering
	\includegraphics[width=9cm]{comparison_of_three_methods.png}
	\caption{Information rates for different encoding methods for $n = 127$.} 
	%\hm{Let's consider using a table with some selected values of $\gamma$. Currently, we cannot see the difference the rates between ${\rm ENC}^{J}$ and ${\rm ENC}^{K}$. }}{\color{red} JC : Do you mean the difference between ${\rm ENC}^{\vI}$ and ${\rm ENC}^{\vK}$?.}
\label{fig:threecomparison}
\end{figure}

	\end{comment}
	
	\begin{table}
		\centering
		\begin{tabular}{|c|c|c|c|}
			\hline & & &\\[-10pt]
			$\gamma$, where $T = \gamma n$ & ${\rm ENC}^{\vA}$ & ${\rm ENC}^{\vB}$ & ${\rm ENC}^{\vC}$ \\[1pt] \hline & & & \\[-11pt]
%			1.0 & 0.000 & & 0.000 \\
			1.1 & 0.000 & & 0.364 \\
%			1.2 & 0.000 & & 0.577 \\
%			1.3 & 0.000 & & 0.710 \\
%			1.4 & 0.000 & & 0.786 \\
			1.5 & 0.315 & & 0.841 \\
%			1.6 & 0.443 & & 0.874 \\
%			1.7 & 0.528 & & 0.886 \\
%			1.8 & 0.604 & & 0.892 \\
			1.9 & 0.646 & & 0.890 \\
%			2.0 & 0.661 & & 0.880 \\
%			2.1 & 0.675 & & 0.863 \\
%			2.2 & 0.687 & & 0.844 \\
			2.3 & 0.698 & & 0.821 \\
%			2.4 & 0.709 & & 0.795 \\
%			2.5 & 0.699 & 0.746 & 0.769 \\
%			2.6 & 0.672 & 0.718 & 0.741 \\
			2.7 & 0.665 & 0.691 & 0.714 \\
%			2.8 & 0.658 & 0.667 & 0.689 \\
%			2.9 & 0.635 & 0.644 & 0.665 \\
%			3.0 & 0.614 & 0.622 & 0.643 \\
			3.1 & 0.594 & 0.602 & 0.622 \\
%			3.2 & 0.576 & 0.583 & 0.603 \\
%			3.3 & 0.558 & 0.566 & 0.585 \\
%			3.4 & 0.542 & 0.549 & 0.567 \\
			3.5 & 0.526 & 0.533 & 0.551 \\
%			3.6 & 0.512 & 0.518 & 0.536 \\
%			3.7 & 0.498 & 0.504 & 0.521 \\
%			3.8 & 0.485 & 0.491 & 0.508 \\
			3.9 & 0.472 & 0.479 & 0.495 \\
			4.0 & 0.461 & 0.467 & 0.482 \\ \hline
		\end{tabular}
	
	\vspace{2mm}
		\caption{Information rates for different encoding methods when $n = 127$. \vspace{-7mm}}
		\label{tab:threecomparison}
	\end{table}

	\section{Polynomial-Time Encoders}\label{sec:encoders}

	\subsection{Block encoder ${\rm ENC}^\vA$}\label{subsec:generalencoder}
	
	In 1984, Tenengolts \cite{Tenengolts1984} presented a systematic encoder that maps quaternary sequences to an $(n,4n,\B)$-synthesis code with at most $\ceil{\log_2 n} +12$ bits of redundancy. With a little modification of the encoder from Tenengolts, we provide a \textit{systematic single indel-correcting encoder} that runs in linear time with only $\ceil{\log_2 n} +6$ bits of redundancy. We remark that a non-systematic linear-time encoder with $\ceil{\log_2 n}+2$ redundant bits were given in~\cite{Cai2021}. 
	However, we need a systematic encoder for the purpose of this paper.
	
	\begin{definition}\label{def:indelencoder}
		The map ${\rm ENC} :\Sigma^m \xrightarrow{} \Sigma^n$ is a \textit{systematic single-indel correcting encoder} if there exists a \textit{decoder} map ${\rm DEC} : \Sigma^{n*} \xrightarrow{} \{0,1\}^m$ such that
		\begin{enumerate}[(i)]
			\item \label{condition1}For every $\bx \in \Sigma^m$, if $\bc={\rm ENC}(\bx)$ and $\by \in \B(\bc)$, then ${\rm DEC}(\by)=\bx$,
			\item \label{condition2}There exists an injective map $\omega : \{1,2,\ldots,m\} \xrightarrow{} \{1,2,\ldots,n\}$, such that for every $\bx \in \Sigma^m$, if ${\rm ENC}(\bx)=\bc$, then $x_i=c_{\omega(i)}$ for all $1\leq i \leq m$.
		\end{enumerate}
		This means that the code $\cC=\{{\rm ENC}(\bx): \bx \in \Sigma^m \}$ is a \textit{systematic single-indel correcting code}.
	\end{definition}
	
	In order to have an encoder from binary sequences to quaternary sequences, we use a natural bijection $\phi$ from quaternary alphabet $\Sigma=\{1,2,3,4\}$ to the following two-bit binary sequences:
	$$ 1 \longleftrightarrow 00, \quad 2 \longleftrightarrow 01,\quad  3 \longleftrightarrow 10, \quad 4 \longleftrightarrow 11.$$
	Then, we extend the mapping $\phi$ to map a quaternary sequence $\bx \in \Sigma^n$ to a binary sequence in $\{0,1\}^{2n}$. For example, $\phi(2314)=01100011$.
	
First, we present our \textit{systematic single-indel correcting encoder} in Algorithm \ref{alg:encoder}. The input of this encoder is a quaternary sequence of length $m=n-\ceil{\log_4 n} - 3$, where $n$ is the length of the output quaternary sequence. Thus the redundancy of the encoder in bits is $\ceil{\log_2 n} +6$ bits.
	
	\begin{algorithm}[!h]
		\caption{Systematic single-indel correcting encoder ${\rm ENC}^\vH$}\label{alg:encoder}
		\textbf{Input} $\quad \bx=x_1 x_2 \cdots x_m \in \Sigma^m$ \\
	\textbf{Output} $\bc=c_1 c_2 \cdots c_n={\rm ENC}^\vH(\bx) \in (n,4n,\B)$-synthesis code 
		\begin{algorithmic}[1]
			\STATE Set $c_1 \cdots c_{m}= \bx$.
			\STATE Set $c_{m+1} = c_{m+2}=c_m +2 \bmod 4$.
			\STATE Set $c_{m+3}=\sum_{i=1}^{m}{c_i} \bmod 4$.
			\STATE Set $c_{m+4} \cdots c_{n}$ to be the quaternary representation of ${\rm VT}(\tilde{\bx}) \bmod n$.
			\STATE \textbf{return} $\bc$.
		\end{algorithmic}
	\end{algorithm}
	
	For completeness, we also present our decoder map ${\rm DEC}^\vH:\Sigma^{n*} \xrightarrow{} \Sigma^m$ in Algorithm \ref{alg:decoder}.
	
	\begin{theorem}
		${\rm ENC}^\vH$ is a \textit{systematic single-indel correcting encoder} that maps quaternary sequences of length $m=n-\ceil{\log_4 n} - 3$ to a quaternary single-indel correcting correcting code of length $n$ with $\ceil{\log_2 n} +6$ bits of redundancy. Furthermore, the encoder runs in $O(n)$ time.
		
	\end{theorem}

	\begin{algorithm}[!h]
		\caption{Systematic single-indel correcting decoder ${\rm DEC}^\vH$}\label{alg:decoder}
		\textbf{Input} $\quad \bc \in \Sigma^{n*}$ \\
		\textbf{Output} $\bx \in \Sigma^m ={\rm DEC}^\vH(\bc)$
		\begin{algorithmic}[1]
			\STATE Let $\ell=|\bc|$ denote the length of the input $\bc$.
			\STATE If $\ell=n$, then we \textbf{return} $\bx=c_1 \cdots c_m$.
			\STATE If $\ell=n+1$, then a single insertion happened. If $c_{m+1}=c_{m+2}$, then there is no error in the data part, hence we output $\bx=c_1 \cdots c_m$. Otherwise, set $a=c_{m+4}$ and $b$ to be decimal representation of the quaternary word $c_{m+5} \cdots c_{n+1}$. Using Theorem \ref{thm:indeldecoder}, we \textbf{return} $\bx={\rm DEC}_{a,b}(c_1 \cdots c_{m+1})$
			\STATE If $\ell=n-1$, then a single deletion happened. If $c_{m}=c_{m+1}$, then set $a=c_{m+2}$ and $b$ to be decimal representation of the quaternary word $c_{m+3} \cdots c_{n-1}$. Using Theorem \ref{thm:indeldecoder}, we output $\bx={\rm DEC}_{a,b}(c_1 \cdots c_{m-1})$. Otherwise, we \textbf{return} $\bx=c_1 \cdots c_m$.
		\end{algorithmic}
	\end{algorithm}

We can verify easily that ${\rm ENC}^\vH$ satisfies the definition of \textit{systematic single-indel correcting encoder} in Definition \ref{def:indelencoder}.

Now, we present a \textit{synthesis-constrained encoder} which encodes binary messages to an $(n,T')$-synthesis code using ranking/unranking algorithm. Let $\cS$ be a finite set with cardinality $N$. A \textit{ranking function} for a set $\cS$ is a bijection $\text{rank} :\cS \xrightarrow{} \{1,2,\ldots,N\}$. There is a unique \textit{unranking function} associated with the function rank, which is $\text{unrank}:\{1,2,\ldots,N\} \xrightarrow{} \cS$, so that $\text{rank}(s)=i$ if and only if $\text{unrank}(i)=s$ for all $s \in \cS$ and $i \in \{1,2,\ldots,N\}$. Recall that $W(n,T)$ denote the set of all quaternary words of length $n$ whose synthesis time is at most $T$  and $A(n,T)$ is its size. We define the map
	\begin{equation*}
		\sigma^{n,k}_{T} : W(n-1,T-k) \xrightarrow{} W(n,T) ,
	\end{equation*}
	such that $$\sigma^{n,k}_{T}((x_1,x_2, \ldots, x_{n-1}))=(x_1,\ldots,x_{n-1},(x_{n-1}+k)\bmod 4),$$
	if $(x_1,x_2, \ldots, x_{n-1}) \in W(n-1,T-k)$.
	\begin{theorem}
		The map $\sigma^{n,k}_{T}$ is injective and
		\begin{equation}
			W(n,T)=\bigsqcup_{k=1}^4{\sigma^{n,k}_{T}(W(n-1,T-k))},    \label{recursion}
		\end{equation}
		where $\bigsqcup$ denotes the disjoint union of sets.
	\end{theorem}

	\begin{proof}
		First, we prove injectivity. If $\bx,\bx' \in W(n-1,T-k)$ where $\bx \neq \bx'$ then $\sigma^{n,k}_{T}(\bx) \neq \sigma^{n,k}_{T}(\bx')$ since the first $n-1$ positions of $\sigma^n_{T}(\bx)$ and $\sigma^n_{T}(\bx')$ are different.
		
		Next, we prove that $W(n,T)=\bigsqcup_{k=1}^4{\sigma^{n,k}_{T}(W(n-1,T-k))}$. If $\bx \in W(n-1,T-k)$ and $\bx' \in W(n-1,T-k')$, where $k \neq k'$ but $\bx$ and $\bx'$ are not necessarily distinct, then $\by=\sigma^{n,k}_{T}(\bx) \neq \sigma^{n,k'}_{T}(\bx')=\by'$ since $k=y_n - y_{n-1} \bmod 4 \neq y'_n - y'_{n-1}\bmod 4=k'$. Thus, the four sets are disjoint. Furthermore, note that for any $\bx \in W(n-1,T-k)$, we have $S(\sigma^{n,k}_{T}(\bx))=||D(\sigma^{n,k}_{T}(\bx))||=||D(\bx)||+k$. Thus for any $\bx \in W(n-1,T-k)$, we have $\sigma^{n,k}_{T}(\bx) \in W(n,T)$ and hence $\sigma^{n,k}_{T}(W(n-1,T-k)) \subset W(n,T)$. Next, suppose that $\by=(y_1,y_2,\ldots,y_n) \in W(n,T)$. Let $k=y_n - y_{n-1} \bmod 4$. Let $\bx=(y_1,y_2,\ldots,y_{n-1})$. Note that $\bx \in W(n-1,T-k)$, and $\sigma^{n,k}_{T}(\bx)=\by$. Thus, $W(n,T) \subset \bigsqcup_{k=1}^4{\sigma^{n,k}_{T}(W(n-1,T-k))}$.
	\end{proof}

	\begin{corollary}\label{cor:ranking}
		$ A(n,T)=\sum_{k=1}^{4}{A(n-1,T-k)}.$
	\end{corollary}
	\begin{proposition}\label{prop:base}
		The base cases of the recursion in \eqref{recursion} are
		\begin{itemize}
			\item $W(1,T)=\{1,\ldots,T \}$, if $1\leq T \leq 4$,
			\item $W(1,T) =\{1,2,3,4 \}$, if $T>4$,
			\item $W(n,T) =\emptyset$, if $T\leq 0$ and $n\geq 1$.
		\end{itemize}
	\end{proposition}
	Thus from~\eqref{recursion} and Proposition~\ref{prop:base}, we can build $W(n,T)$ recursively from the base cases.
	
	\begin{example}\label{ex:ranking}
		We can construct $W(3,6)$ from $W(2,5), W(2,4),W(2,3)$ and $W(2,2)$ using the map $\sigma^{3,k}_6$. Then $W(2,5)$ can in turn be constructed from $W(1,4)=\{1,2,3,4\}$, $W(1,3)=\{1,2,3\}$, $W(1,2)=\{1,2\}$, and $W(1,1)=\{1\}$ using the map $\sigma^{2,k}_5$. Similarly for $W(2,4),W(2,3)$ and $W(2,2)$.
		
		Thus, we obtain 
		\begin{align*}
			W(2,5)&=\sigma^{2,1}_{5}(W(1,4)) \cup \sigma^{2,2}_{5}(W(1,3))\\
			&\cup \sigma^{2,3}_{5}(W(1,2)) \cup \sigma^{2,4}_{5}(W(1,1))\\
			&=\{12,23,34,41\} \cup \{13,24,31\} \cup \{14,21\}\cup \{11\}.
		\end{align*}
		\begin{align*}
			W(2,4)&=\sigma^{2,1}_{4}(W(1,3)) \cup \sigma^{2,2}_{4}(W(1,2))\cup \sigma^{2,3}_{4}(W(1,1)) \\
			&=\{12,23,34\} \cup \{13,24\} \cup \{14\}.
		\end{align*}
		\begin{align*}
			W(2,3)&=\sigma^{2,1}_{3}(W(1,2)) \cup \sigma^{2,2}_{3}(W(1,1))\\
			&=\{12,23\} \cup \{13\}.
		\end{align*}
		\begin{align*}
			W(2,2)&= \sigma^{2,1}_{2}(W(1,1)) \\
			&=\{12\}.
		\end{align*}
		And finally, we have
		\begin{align}
			W(3,6)&=\sigma^{3,1}_{6}(W(2,5)) \cup \sigma^{3,2}_{6}(W(2,4))\nonumber\\
			&\cup \sigma^{3,3}_{6}(W(2,3)) \cup \sigma^{3,4}_{6}(W(2,2)) \nonumber\\
			&=\{123,234,341,412,134,241,312,141,212,112\} \nonumber \\
			&\cup \{124,231,342,131,242,142\} \nonumber \\
			&\cup \{121,232,132 \}\ \cup \{122 \}\label{eq:w38}.
		\end{align}
	\end{example}

In fact, they give rise to a polynomial time algorithm, via dynamic programming, to build $W(n,T)$ and to calculate its size. In this paper, we focus
	on explaining the ideas behind the design of our ranking/unranking algorithms and not the optimization of the detailed running time analysis of these algorithms. The idea of ranking/unranking algorithm is the unfolding of the recursion given in \eqref{recursion}. We can follow the natural \textit{total ordering} of the elements in $W(n,T)$ where the elements in $\sigma^{n,k}_T(W(n-1,T-k))$ are before after the elements in $\sigma^{n,k'}_T(W(n-1,T-k'))$ if $k<k'$. 
	Our goal is to propose a polynomial time ranking/unranking algorithm for the set $W(n,T)$. The formal unranking algorithm is presented in Algorithm~\ref{alg:unranking}.
	
	\begin{algorithm}[!h]
		\caption{$\text{unrank}(j,n,T)$}\label{alg:unranking}
		\textbf{Input} Integers $j;n,T$, where $1 \leq j \leq A(n,T).$ \\
		\textbf{Output} The $j$-th element of $W(n,T)$.\\[-5mm]
		\begin{algorithmic}[1]
			\STATE If $n=1$ and $1 \leq j\leq T\leq 4$, then \textbf{return} $(j)$.
			\STATE If $n=1$, $T>4$ and $1 \leq j\leq 4$, then \textbf{return} $(j)$.
			\STATE Let $k \leq 4$ be the smallest positive integer such that $j \leq \sum_{i=1}^k{A(n-1,T-k)}$.
			\STATE Let $S=\sum_{i=1}^{k-1}{A(n-1,T-k)}$ if $k>1$, and let $S=0$ if $k=1$.
			\STATE \textbf{return} $\sigma^{n,k}_T(\text{unrank}(j-S;n-1,T-k))$
		\end{algorithmic}
	\end{algorithm}

	The corresponding ranking algorithm can be defined analogously in Algorithm~\ref{alg:ranking}.
	
	\begin{algorithm}[!h]
		\caption{$\text{rank}(\bx,n,T)$}\label{alg:ranking}
		\textbf{Input} Integers $n,T$ and string $\bx \in W(n,T)$.\\
		\textbf{Output} Integer $j$ s.t. $\bx$ is the $j$-th element of~$W(n,T)$.\\[-5mm]
		\begin{algorithmic}[1]
			\STATE If $n=1$ then \textbf{return} $\bx$.
			\STATE Let $k=x_n - x_{n-1} \bmod 4$.
			\STATE Let $S=\sum_{i=1}^{k-1}{A(n-1,T-i)}$ if $k>1$, and let $S=0$ if $k=1$.
			\STATE \textbf{return} $S+\text{rank}((x_1,x_2,\ldots,x_{n-1}),n-1,T-k)$.
		\end{algorithmic}
	\end{algorithm}

	Now, we are ready to present our \textit{single-indel correcting synthesis-constrained encoder} for general $T$. We assume that the input binary sequence $\bx$ consists of $\ell$ blocks, each of length $m$, where the $i$-th block is denoted by $\bx^{(i)}$, for $1\leq i \leq \ell$. Each block of $\bx^{(i)}$ is converted to its decimal representation $c^{(i)}$, which is then mapped to a quaternary string $\by^{(i)}$ of length $k$ in $(k,T'/\ell)$-synthesis code using $\text{unrank}$ function. Note that we must have $2^m \leq A(k,T'/\ell)$, and therefore we set $m =\floor{\log_2 A(k,T'/\ell)}$. Finally, we concatenate all the blocks $\by^{(i)}$'s and encode it to a quaternary string in $(n,T,\B)$-synthesis code using ${\rm ENC}^\vH$. Here $T=T'+\epsilon$, where $\epsilon=4(\ceil{\log_4 n}+\ell+3)$, where the additional $4\ell$ might come from the concatenated parts.
	%%or $\epsilon=4(\ceil{\log_4 n}+3)+3(\ell-1)$
	Therefore, if the length of the output quaternary sequence is $n$, then $\ell k=n - \log_4 n- 3$. Thus the redundancy of the encoder in bits is $2n -\ell m=2n - \ell \floor{\log_2 A(\frac{n - \log_4 n- 3}{\ell},\frac{T-\epsilon}{\ell})}$ bits. The formal encoder is described in Algorithm~\ref{alg:constrainedencoderquad} and the result is summarized in Theorem~\ref{thm:generalencoder}.
	
	\begin{algorithm}[!h]
		\caption{Single-indel correcting synthesis-constrained encoder ${\rm ENC}^\vA(\ell,m,\bx;n,T)$}\label{alg:constrainedencoderquad}
		\textbf{Input} $\bx=x_1 x_2 \cdots x_{\ell m} \in \{0,1\}^{\ell m}$ such that $2^m \leq A(\frac{n - \log_4 n- 3}{\ell},\frac{T-\epsilon}{\ell})$, where $\epsilon=4(\ceil{\log_4 n}+\ell+3)$.\\
		\textbf{Output} $\bc=c_1 c_2 \cdots c_{ n} ={\rm ENC}^\vA(\ell,m,\bx;n,T) \in (n,T,\B)$-synthesis code 
		\begin{algorithmic}[1]
			\STATE Set $\bx^{(t)}=x_{m(t-1)+1}, x_{m(t-1)+2},\ldots,x_{m(t-1)+m}$, for all $1\leq t \leq \ell$ which represents the $t$-th block of $\bx$.
			\STATE Let $c^{(t)}=\sum_{i=1}^m {\bx^{(t)}_i 2^{\ell-i}}$ be the decimal representation of the binary string $\bx^{(t)}$.
			\STATE Let $\by^{(t)}=\text{unrank}(c^{(t)};\frac{n - \log_4 n- 3}{\ell},\frac{T-\epsilon}{\ell})$.
			\STATE Set $\by=\by^{(1)}\by^{(2)}\cdots \by^{(\ell)} \in \Sigma^{n - \log_4 n- 3}$.
			\STATE Set $\bc= {\rm ENC}^\vH (\by)\in \Sigma^n$ and \textbf{return} $\bc$.
		\end{algorithmic}
	\end{algorithm}
	\begin{theorem}\label{thm:generalencoder}
		The encoder ${\rm ENC}^\vA$ is a \textit{single-indel correcting synthesis-constrained encoder} that maps a binary string $\bx$ of length $\ell m$ to a quaternary string of length $n$ with synthesis time at most $T$. This encoder runs in $O(n^3)$ time with $2n - \ell \floor{\log_2 A(\frac{n - \log_4 n- 3}{\ell},\frac{T-\epsilon}{\ell})}$ bits of redundancy, where $\epsilon=4(\ceil{\log_4 n}+\ell+3)$. Thus the code $\cC=\{{\rm ENC}^\vA(\ell,m,\bx;n,T): \bx \in \{0,1\}^{\ell m \}}$ is an $(n,T,\B)$-synthesis code for a fixed $n<T \leq 4n$.
	\end{theorem}

	\begin{proof}
		We want to show that the synthesis time of the output $\bc$ of encoder ${\rm ENC}^\vA$ is indeed at most $T$. Note that the synthesis time of each $\by^{(t)}$ is at most $\frac{T-\epsilon}{\ell}$. Thus the synthesis time of $\by$ is at most $T-\epsilon + 4 \ell$, where the additional $4\ell$ synthesis time may come from the concatenated parts between each block. Recall that $\bc= {\rm ENC}^\vH (\by)$ and ${\rm ENC}^\vH$ is a systematic encoder. Furthermore, the encoder appends a quaternary substring of length $\ceil{\log_4 n}+3$, which increases the synthesis time by at most $4(\ceil{\log_4 n}+3)$. Thus $S(\bc) \leq T- \epsilon +4\ell +4(\ceil{\log_4 n}+3) = T$.
	\end{proof}

	The comparison of the information rates for the encoder ${\rm ENC}^\vA$ for different values of block lengths $k$ with fixed output length $n=127$ is given in Figure~\ref{fig:algo1}.
	
	For completeness, we also provide our decoder map ${\rm DEC}^\vA:\Sigma^{n*} \xrightarrow{} \Sigma^m$ in Algorithm \ref{alg:decoderA}.
	
	\begin{algorithm}[!h]
		\caption{Single-indel correcting synthesis-constrained decoder ${\rm DEC}^\vA(\ell,m;n,T,\bc)$}\label{alg:decoderA}
		\textbf{Input} $\bc \in \Sigma^{n*}.$\\
		\textbf{Output} $\bx=x_1 x_2 \cdots x_{ \ell m} ={\rm DEC}^\vA(\ell,m;n,T,\bc) $.
		\begin{algorithmic}[1]
			\STATE Set $\by={\rm DEC}^\vH(\bc) \in \Sigma^{n - \log_4 n -3}$.
			\STATE Split $\by$ into $\ell$ blocks of equal length, such that $\by=\by^{(1)}\by^{(2)}\cdots \by^{(\ell)}$.
			\STATE Let $c^{(t)}=\text{rank}(\by^{(t)},\frac{n - \log_4 n- 3}{\ell},\frac{T-\epsilon}{\ell})$, where $\epsilon=4(\ceil{\log_4 n}+\ell+3)$, for all $1 \leq t \leq \ell$.
			\STATE Let $\bx^{(t)}\in \{0,1\}^m$ be the binary representation of length $m$ of $c^{(t)}$.
			\STATE Set $\bx$ to be the concatenation of $\bx^{(1)}\bx^{(2)} \cdots \bx^{(\ell)} \in \{0,1\}^{\ell m}$ and \textbf{return} $\bx$.
		\end{algorithmic}
	\end{algorithm}

	\begin{figure}
		\label{fig:algo1}
		\centering
		\includegraphics[width=10cm]{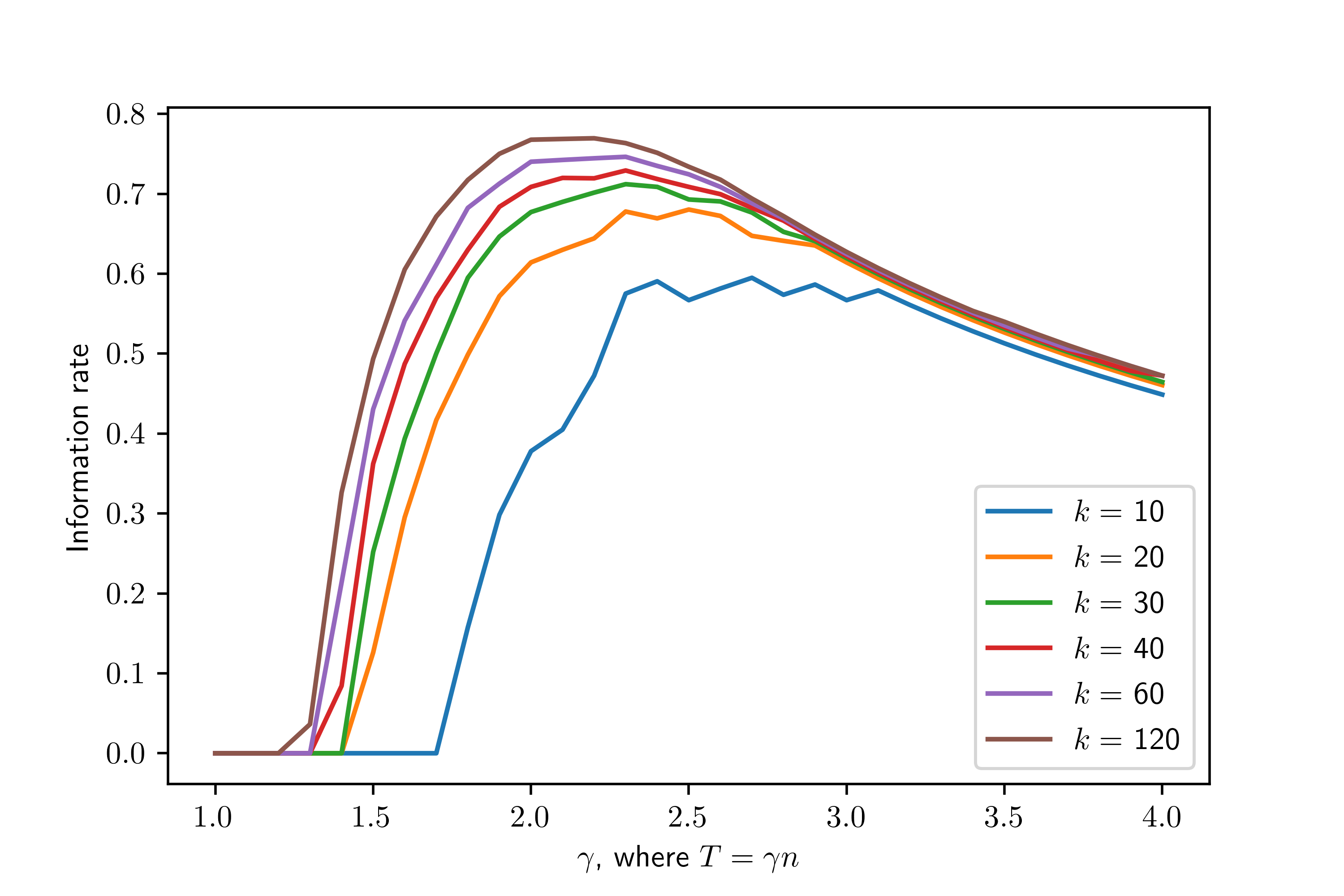}
		\caption{Information rate of block encoder ${\rm ENC}^\vA$ for $n = 127$ with multiple block size values \vspace{-3mm}}
		\label{fig:algo1}
	\end{figure}
	\subsection{Special encoder ${\rm ENC}^\vB$ for $T \geq 2.5 n$.}
	\label{subsec:special} 
	For the special case of $T \geq 2.5 n$, we only need two bits of redundancy when we encode the binary messages to an $(n,T')$-synthesis code, where $T=T'+O(\log n)$. And then similar to the general encoder presented in the previous subsection, we use \textit{systematic single-indel correcting encoder}, where the extra redundant bits are again encoded to minimize the synthesis time. This special encoder ${\rm ENC}^\vB$ is presented in Algorithm \ref{alg:constrainedencoder}. For simplicity, we assume that the length of the input sequence is even. The input of this encoder is a binary sequence of length $2m=2n-2\ceil{\log_4 n} - 10$, where $n$ is the length of the output quaternary sequence. Thus the redundancy of the encoder ${\rm ENC}^\vB$ is $\ceil{\log_2 n} +10$ bits.
	
	\begin{algorithm}[!h]
		\caption{Single-indel correcting synthesis-constrained encoder ${\rm ENC}^\vB(m,\bx;n)$ for $T \geq 2.5n$}\label{alg:constrainedencoder}
		\textbf{Input} $\quad \bx=x_1 x_2 \ldots x_{2m} \in \{0,1\}^{2m}$ \\
			\textbf{Output} $\bc=c_1 c_2 \ldots c_n={\rm ENC}^\vB(m,\bx;n) \in (n,T,\B)$-synthesis code 
		\begin{algorithmic}[1]
			\STATE Set $\by=\phi^{-1}(\bx)$ to be a quaternary word of length $m$.
			\STATE Let $\by' \in \Sigma^{m+1}$ such that $y'_1 \cdots y'_m=D(\by)$ is the \textit{differential word} of $\by$, and $y'_{m+1}=1$.
			\STATE \label{state:3}If $||\by'||> 2.5 (m+1)$, then for all $1 \leq i \leq m+1$, replace $y_i$ with $5 - y_i$.
			\STATE \label{state:4}Set $\bz= {\rm ENC}^\vH (D^{-1}(\by'))$, where $z_1 \cdots z_{m+1}=D^{-1}(\by')$ is the systematic part of the encoding and $z_{m+2} \cdots z_{n-1}$ is the suffix of $\bz$.
			\STATE Let $\bz'$ be a sequence of length $n$ where $z'_1 \cdots z'_{n-1}=D(\bz)$, and $z'_n =1$.
			\STATE If $||z'_{m+2} \cdots z'_{n}||>2.5 (n-m-1)$, then replace $z'_i$ with $5-z'_i$ for all $m+2 \leq i \leq n$. \label{state:6}
			\STATE Let $\bc=D^{-1}(\bz')$ and \textbf{return} $\bc$.
		\end{algorithmic}
	\end{algorithm}

	For completeness, we also provide the nontrivial decoder ${\rm DEC}^\vB$ for ${\rm ENC}^\vB$ in Algorithm \ref{alg:constraineddecoder}.
	
	\begin{algorithm}[!h]
		\caption{Single-indel correcting synthesis-constrained decoder ${\rm DEC}^\vB(m;n,\bc)$ for $T \geq 2.5n$}\label{alg:constraineddecoder}
		\textbf{Input} $\quad \bc \in \Sigma^{n*}$.\\
		\textbf{Output} $\bx \in \{0,1\}^{2m}$
		\begin{algorithmic}[1]
			\STATE Let $\ell=|\bc|$ denote the length of the input $\bc$.
			\STATE If $\ell=n+1$ and $c_{m+2}\neq c_{m+3}$, then the insertion happened at the first $m+3$ positions. Set $\bc'=D(\bc)$. If $c'_{n+1}=4$, replace $c'_i$ with $5-c'_i$ for all $m+5 \leq i \leq n+1$ and set $\bz=D^{-1}(c'_1 c'_2 \cdots c'_{n})$. Set $a=z_{m+5}$ and $b$ to be decimal representation of the quaternary word $z_{m+6} \cdots z_{n}$. Using Theorem \ref{thm:indeldecoder}, we compute $\by={\rm DEC}_{a,b}(z_1 \cdots z_{m+2})$.
			\STATE If $\ell=n+1$ and $c_{m+2}= c_{m+3}$, then there is no error in the first $m+1$ positions. Set $\by=c_1 c_2 \cdots c_{m+1}$.
			\STATE If $\ell=n-1$ and $c_{m+1}=c_{m+2}$, then the deletion happened at the first $m+1$ positions. Set $\bc'=D(\bc)$. If $c'_{n-1}=4$, replace $c'_i$ with $5-c'_i$ for all $m+1 \leq i \leq n-1$ and set $\bz=D^{-1}(c'_1 c'_2 \cdots c'_{n-2})$. Set $a=z_{m+3}$ and $b$ to be decimal representation of the quaternary word $z_{m+4} \cdots z_{n-2}$. We compute $\by={\rm DEC}_{a,b}(z_1 \cdots z_{m})$.
			\STATE If $\ell=n-1$ and $c_{m+1}\neq c_{m+2}$, then there is no error in the first $m+1$ positions. Set $\by=c_1 c_2 \cdots c_{m+1}$.
			\STATE If $\ell=n$, then set $\by=c_1 c_2 \cdots c_{m+1}$.
			\STATE Set $\by'=D(\by)$. If $y'_{m+1}=4$, replace $y'_i$ with $5- y'_i$ for all $1\leq i \leq m$. Set $\bx=\phi(D^{-1}(y'_1 y'_2 \cdots y'_m))$ and \textbf{return} $\bx$.
		\end{algorithmic}
	\end{algorithm}

The result is formalized in Theorem \ref{thm:constrainedencoder2.5}.
	\begin{theorem}\label{thm:constrainedencoder2.5}
		The encoder ${\rm ENC}^\vB$ is a \textit{single-indel correcting synthesis-constrained encoder} that maps a binary string $\bx$ of length $2m$ to a quaternary string of length $n$ with synthesis time at most $2.5 n$. This encoder runs in linear time with $10+ \ceil{\log_2 n}$ bits of redundancy. Thus the code $\cC=\{{\rm ENC}^\vB(m,\bx;n): \bx \in \{0,1\}^{2m} \}$ is an $(n,T,\B)$-synthesis code for $T \geq 2.5 n$.
	\end{theorem}

	\begin{proof}
		Let $\bx \in \{0,1\}^{2m}$ be the input sequence. Note that $\by=\phi^{-1}(\bx)$ is the quaternary representation of the input binary string. Observe that Step \ref{state:3} guarantees that the synthesis time of $D^{-1}(\by')$ is $||D(D^{-1}(\by'))||=||\by'|| \leq 2.5 (m+1)$.

		Furthermore, Step~\ref{state:6} guarantees that the suffix of $\bz$ after the systematic part of $z$ in Step~\ref{state:4} also has synthesis time at most $2.5 (n-m-1)$. Thus the total synthesis time for the systematic part and the suffix of $\bz$ after Step~\ref{state:6} is at most $2.5 n$.
		Finally, we can verify that ${\rm ENC}^\vB$ satisfies both conditions in Definition~\ref{def:constrainedencoder}.
	\end{proof}

\subsection{Direct encoder ${\rm ENC}^\vC$} \label{subsec:directencoder}
	Finally, in this subsection, we provide ranking/unranking algorithm that encodes binary messages directly into a ${\rm VT}_n(a,b,T)$ code for any $n<T\leq 4n$. This encoder should give the best redundancy among other encoders, however the drawback is that this encoder takes longer time to run compared to the previous encoders. Similar to Subsection~\ref{subsec:generalencoder}, we need to introduce the following recursion.
	
	Let $W_n(\ell,T,a,b,\alpha)$ denote the set of all quaternary words $\bx=x_1 x_2 \cdots x_\ell$ of length $\ell$, such that $x_\ell=\alpha$, $S(\bx)\leq T$ and ${\rm VT}(\tilde{\bx})=a \bmod n, \text{ and } \sum_{i=1}^\ell{\bx_i}=b \bmod 4$. Note that $b$ and $\alpha$ are always in $\{1,2,3,4\}$, namely the shifted modulo $4$, and $a$ is always \text{mod} $n$, where $n$ is fixed in the beginning. We have the following recursion.
	\begin{align}
		&W_n(\ell,T,a,b,\alpha) \nonumber \\
		&= \bigsqcup_{i=1}^{4} W_n(\ell-1,T-c,a- \mathbbm{1}_{\alpha \geq i} (\ell-1),b-\alpha,i ) \circ \alpha ,\label{eq:recdirect}
	\end{align}
	where $c=\alpha-i \bmod 4$ and $\ell \geq 2$.
	For example, we have $W_n(\ell,T,a,b,3) = W_n(\ell-1,T-2,a-\ell+1,b-3,1 ) \circ 3  \sqcup W_n(\ell-1,T-1,a  -\ell+1,b-3,2)\circ 3\sqcup W_n(\ell-1,T-4,a -\ell+1,b-3,3)\circ 3\sqcup W_n(\ell-1,T-3,a,b-3,4 )\circ 3$.
	\begin{proposition}\label{prop:basedirect}
		The base cases of the recursion in \eqref{eq:recdirect} are
		\begin{itemize}
			\item $W_n(1,T,a,b,\alpha)=\emptyset$, if $a \neq 0$ or $\alpha \neq b$ or $\alpha>T$,
			\item $W_n(1,T,0,\alpha,\alpha)=\{\alpha\}$, if $\alpha \leq T$
		\end{itemize}
	\end{proposition}
	Similar to Subsection~\ref{subsec:generalencoder}, using ranking/unranking algorithm, from~\eqref{eq:recdirect} and Proposition~\ref{prop:basedirect}, we can build $W_n(n,T,a,b,*)\triangleq\bigsqcup_{\alpha=1}^{4}{W_n(n,T,a,b,\alpha)}$ recursively from the base cases.
	The formal unranking algorithm is presented in Algorithm~\ref{alg:unranking2}, in which our goal is to find the $j$-th element of $W_n(n,T,a,b,*)$.
	
\begin{algorithm}[!h]
	\caption{$\text{unrank}(j;n,\ell,T,a,b,\alpha)$}\label{alg:unranking2}
	\textbf{Input} Integers $j,n,\ell,T,a,b,\alpha$, where $\alpha \in \mathbb{Z}_4 \cup \{*\}, b \in\mathbb{Z}_4, 1\leq \ell\leq n, a \in \mathbb{Z}_n$ and $1 \leq j \leq |W_n(\ell,T,a,b,\alpha)|.$ \\
	\textbf{Output} The $j$-th element of the set $W_n(\ell,T,a,b,\alpha)$.
	\begin{algorithmic}[1]
		\STATE If $\ell=j=1$, $b=\alpha \leq T$, and $a=0$ then \textbf{return} $(\alpha)$.
		\STATE If $\ell=n$ and $\alpha=*$, then let $\beta \leq 4$ be the smallest positive integer such that $j \leq \sum_{i=1}^\beta{|W_n(n,T,a,b,i)|}$. If $\beta>1$, then set $S=\sum_{i=1}^{\beta-1}{|W_n(n,T,a,b,i)|}$ , otherwise, if $\beta=1$ then set $S=0$. \textbf{return} unrank$(j-S;n,n,T,a,b,\beta)$.
		\STATE If $\alpha\neq *$, then let $k \leq 4$ be the smallest positive integer such that $j \leq \sum_{i=1}^k{|W_n(\ell-1,T-c,a- \mathbbm{1}_{\alpha \geq i} (\ell-1),b-\alpha,i)|}$, where $c=\alpha-i \bmod 4$.
		\STATE If $k>1$, then let $S=\sum_{i=1}^{k-1}{|W_n(\ell-1,T-c,a- \mathbbm{1}_{\alpha \geq i} (\ell-1),b-\alpha,i)|}$, where $c=\alpha-i \bmod 4$. Otherwise, if $k=1$ set $S=0$.
		\STATE \textbf{return} unrank$(j-S; n,\ell -1,T-c,a- \mathbbm{1}_{\alpha \geq j} (\ell-1),b-\alpha,j ) \circ \alpha$, where $c=\alpha-k \bmod 4$.
	\end{algorithmic}
\end{algorithm}

	The corresponding ranking algorithm can be defined analogously in Algorithm~\ref{alg:ranking2}.
		\begin{algorithm}[!h]
		\caption{$\text{rank}(\bx,n,\ell,T,a,b,\alpha)$}\label{alg:ranking2}
		\textbf{Input} $\bx \in W_n(\ell,T,a,b,\alpha)$ and integers $n,\ell,T,a,b,\alpha$, where $\alpha \in \mathbb{Z}_4 \cup \{*\}, b \in\mathbb{Z}_4, 1\leq \ell\leq n, a \in \mathbb{Z}_n$ \\
		\textbf{Output} Integer $j$ s.t. $\bx$ is the $j$-th element of $W_n(\ell,T,a,b,\alpha)$.
		\begin{algorithmic}[1]
			\STATE If $\ell=1$ and $\alpha = *$ then {\bf return} $\alpha$, else if $\ell=1$ and $\alpha \neq *$, then $ {\bf return}$ $1$.
			\STATE If $\alpha=*$, then let $\beta=x_\ell$. If $\beta>1$, then let $S=\sum_{i=1}^{\beta -1}{|W_n(\ell,T,a,b,i)|}$, otherwise, if $\beta=1$, then $S=0$. \textbf{return } $S+\text{rank}(\bx,n,\ell,T,a,b,\beta)$
			\STATE If $\alpha \neq *$, then let $\beta =x_{\ell -1}$. If $\beta >1$ then let $S= \sum_{i=1}^{\beta -1}{|W_n(\ell-1,T-c,a- \mathbbm{1}_{\alpha \geq i} (\ell-1),b-\alpha,i)|}$, where $c=\alpha-i \bmod 4$, otherwise if $\beta=1$, then let $S=0$. \textbf{return } $S+\text{rank}(x_1 x_2 \cdots x_{\ell -1},n,\ell-1,T-c,a- \mathbbm{1}_{\alpha \geq \beta} (\ell-1),b-\alpha,\beta)$, where $c=\alpha-\beta \bmod 4$.
		\end{algorithmic}
	\end{algorithm}

	Now, we are ready to present the direct encoder ${\rm ENC}^\vC$ in Algorithm~\ref{alg:directencoder}. The input of this encoder is a binary string of length $m \leq |{\rm VT}_n(a,b,T)|$ and the output is a quaternary sequence in ${\rm VT}_n(a,b,T)$.
	
	\begin{algorithm}[!h]
		\caption{Direct encoder ${\rm ENC}^\vC(m,\bx;n,T,a,b)$}\label{alg:directencoder}
		\textbf{Input} $\bx=x_1 x_2 \cdots x_{m} \in \{0,1\}^{m}$ such that $2^m \leq |{\rm VT}_n(a,b,T)|$\\
		\textbf{Output} $\bc=c_1 c_2 \cdots c_n={\rm ENC}^\vC(m,\bx;n,T,a,b) \in \text{VT}_n(a,b,T)$ 
		\begin{algorithmic}[1]
			\STATE Let $j=\sum_{i=1}^m {x_i 2^{m-i}}$ be the decimal representation of the binary string $\bx$.
			\STATE \textbf{return} $\text{unrank}(j;n,n,T,a,b,*)$
		\end{algorithmic}
	\end{algorithm}
	
	For completeness, we also provide our decoder map ${\rm DEC}^\vC$ for ${\rm ENC}^\vC$ in Algorithm \ref{alg:directdecoder}.
	
	\begin{algorithm}[!h]
		\caption{Direct decoder ${\rm DEC}^\vC(m;n,T,a,b,\bc)$}\label{alg:directdecoder}
		\textbf{Input} $\bc \in \Sigma^{n*}$\\
		\textbf{Output} $\bx=x_1 x_2 \cdots x_m={\rm DEC}^\vC(m;n,T,a,b,\bc) )$ 
		\begin{algorithmic}[1]
			\STATE Let $\by={\rm DEC}_{a,b}(\bc) \in \Sigma^n$
			\STATE Let $j = \text{rank}(\by,n,n,T,a,b,*)$
			\STATE Let $\bx$ be the binary representation of length $m$ of $j$, and \textbf{return} $\bx$.
		\end{algorithmic}
	\end{algorithm}
The result is formalized in the Theorem~\ref{thm:directencoder}.
	\begin{theorem}
		\label{thm:directencoder}
		The encoder ${\rm ENC}^\vC$ is a \textit{single-indel correcting synthesis-constrained encoder} that maps a binary string $\bx$ of length $m$ to a quaternary codeword in ${\rm VT}_n(a,b,T)$. Furthermore, ${\rm ENC}^\vC$ runs in $O(n^5)$ time and there exist $a \in \mathbb{Z}_n$ and $b \in \mathbb{Z}_4$ such that this encoder has at most $2n-\floor{\log_2 A(n,T)}+\log_2 n + 2 $  bits of redundancy. Thus the code $\cC=\{{\rm ENC}^\vC(m,\bx;n,T,a,b): \bx \in \{0,1\}^m \}$ is an $(n,T,\B)$-synthesis code for a fixed $n<T \leq 4 n$.
	\end{theorem}
	
		\section*{Acknowledgement}
		
		This research of Han Mao Kiah is supported by the Ministry of Education, Singapore, under its MOE AcRF Tier 2 Award MOE-T2EP20121-0007.

\end{document}